\documentclass[a4paper,UKenglish,cleveref]{lipics-v2019}

\usepackage{todonotes}

\newcommand\Class[1]{%
  \mathchoice%
  {\text{\normalfont\fontsize{10pt}{10pt}\selectfont$\mathsf{#1}$}}%
  {\text{\normalfont\fontsize{10pt}{10pt}\selectfont$\mathsf{#1}$}}%
  {\text{\normalfont$\mathsf{#1}$}}%
  {\text{\normalfont$\mathsf{#1}$}}%
}

\usepackage{acro}

\newtheorem{fact}[theorem]{Fact}
{%
  \leavevmode\nobreak\par
	\begin{list}%
		{}%
		{%
			\def\labelstyle{\itshape}
			\setlength{\topsep}{0pt}%
			\settowidth{\labelwidth}{\labelstyle Parameter:}%
			\setlength{\leftmargin}{\labelwidth}%
			\addtolength{\leftmargin}{\labelsep}%
			\setlength{\itemsep}{0pt}%
			\setlength{\parsep}{0pt}%
		}%
	}%
	{%
	\end{list}%
}

\usepackage{tikz}
\usetikzlibrary{positioning}
\usetikzlibrary{calc}
\usetikzlibrary{decorations.pathmorphing}

\usepackage{acro}

\DeclareAcronym{APTAS}{
    short = \textsc{APTAS}, 
    long=asymptotic polynomial time approximation scheme,
    class=abbrev
}

\DeclareAcronym{FPT}{
    short = \textsc{FPT}, 
    long=\textsc{Fixed Parameter Tractable},
    class=abbrev
}

\DeclareAcronym{3HS}{
    short = \textsc{3-HS}, 
    long=\textsc{3-Hitting Set},
    class=abbrev
}

\DeclareAcronym{p3HS}{
    short = \textsc{$p$-3-HS}, 
    long=\textsc{3-Hitting Set},
    class=abbrev
}

\DeclareAcronym{VP}{
    short = \textsc{VP}, 
    long=\textsc{Vector Packing},
    class=abbrev
}

\DeclareAcronym{pVP}{
    short = \textsc{$p$-VP}, 
    long=\textsc{Vector Packing},
    class=abbrev
}

\DeclareAcronym{VMKP}{
	short = \textsc{VMKP}, 
	long=\textsc{Vector Multiple Knapsack},
	class=abbrev
}

\DeclareAcronym{pVMKP}{
	short = \textsc{$p$-VMKP}, 
	long=\textsc{Vector Multiple Knapsack},
	class=abbrev
}

\DeclareAcronym{MKP}{
	short = \textsc{MKP}, 
	long=\textsc{Multiple Knapsack},
	class=abbrev
}

\DeclareAcronym{pMKP}{
	short = \textsc{$p$-MKP}, 
	long=\textsc{Multiple Knapsack},
	class=abbrev
}

\DeclareAcronym{VC}{
    short = \textsc{VC}, 
    long=\textsc{Vector Covering},
    class=abbrev
}

\DeclareAcronym{pVC}{
    short = \textsc{$p$-VC}, 
    long=\textsc{Vector Covering},
    class=abbrev
}

\DeclareAcronym{BP}{
    short = \textsc{BP}, 
    long=\textsc{Bin Packing},
    class=abbrev
}

\DeclareAcronym{pBP}{
    short = \textsc{$p$-BP}, 
    long=\textsc{Bin Packing},
    class=abbrev
}

\DeclareAcronym{BC}{
    short = \textsc{BC}, 
    long=\textsc{Bin Covering},
    class=abbrev
}

\DeclareAcronym{pBC}{
    short = \textsc{$p$-BC}, 
    long=\textsc{Bin Covering},
    class=abbrev
}

\DeclareAcronym{POTRM}{
    short = \textsc{PORM}, 
    long=\textsc{Perfect Over-the-Rainbow Matching},
    class=abbrev
}

\DeclareAcronym{pPOTRM}{
    short = \textsc{$p$-PORM}, 
    long=\textsc{Perfect Over-the-Rainbow Matching},
    class=abbrev
}

\DeclareAcronym{pCM}{
    short = \textsc{$p$-CM}, 
    long=\textsc{Conjoining Matching},
    class=abbrev
}

\DeclareAcronym{CM}{
	short = \textsc{CM}, 
	long=\textsc{Conjoining Matching},
	class=abbrev
}

  \DeclareAcronym{CGM}{
    short = \textsc{CGM}, 
    long=\textsc{Conjoining General Matching},
    class=abbrev
  }

\def\claimedtheorem{}

\newcommand{\NN}{\mathbb{N}}
\newcommand{\QQ}{\mathbb{Q}}
\newcommand{\dotcup}{\ensuremath{\mathop{\dot{\cup}}}}

\newcommand{\vectors}{\mathcal{V}}
\newcommand{\items}{\mathcal{I}}
\newcommand{\profit}[1]{p(#1)}
\newcommand{\colors}{\mathcal{C}}

\newcommand{\cO}{O}

\title{Solving Packing Problems with Few Small Items Using Rainbow Matchings}

\author{Max Bannach}{Institute for Theoretical Computer Science, Universit{\"a}t zu L{\"u}beck, L{\"u}beck, Germany}{bannach@tcs.uni-luebeck.de}{https://orcid.org/0000-0002-6475-5512}{}
\author{Sebastian Berndt}{Institute for IT Security, Universit{\"a}t zu L{\"u}beck, L{\"u}beck, Germany}{s.berndt@uni-luebeck.de}{https://orcid.org/0000-0003-4177-8081}{}
\author{Marten Maack}{Department of Computer Science, Universit{\"a}t Kiel, Kiel, Germany}{mmaa@informatik.uni-kiel.de}{https://orcid.org/0000-0001-7918-6642}{}
\author{Matthias Mnich}{Institut für Algorithmen und Komplexit{\"a}t, TU Hamburg, Hamburg, Germany}{mmnich@tuhh.de}{https://orcid.org/0000-0002-4721-5354}{Supported by DFG grants MN 59/1-1 and MN 59/4-1.}
\author{Alexandra Lassota}{Department of Computer Science, Universit{\"a}t Kiel, Kiel, Germany}{ala@informatik.uni-kiel.de}{https://orcid.org/0000-0001-6215-066X}{Supported by DFG grant ``Strukturaussagen und deren Anwendung in Scheduling- und Packungsprobleme'', JA 612/20-1}
\author{Malin Rau}{Univ. Grenoble Alpes, CNRS, Inria, Grenoble INP, LIG, 38000 Grenoble, France}{malin.rau@inria.fr}{https://orcid.org/0000-0002-5710-560X}{}
\author{Malte Skambath}{Department of Computer Science, Universit{\"a}t Kiel, Kiel, Germany}{malte.skambath@email.uni-kiel.de}{https://orcid.org/0000-0003-2048-3559}{}

\authorrunning{M. Bannach et al.}

\Copyright{Max Bannach, Sebastian Berndt, Marten Maack, Matthias Mnich, Alexandra Lassota, Malin Rau, and Malte Skambath}

\ccsdesc[100]{Theory of computation~Fixed parameter tractability}

\keywords{Bin Packing, Knapsack, matching, fixed-parameter tractable}

\category{}

\supplement{}

\acknowledgements{We want to thank Magnus Wahlstr{\"o}m for helpful discussions.}

\nolinenumbers 

\hideLIPIcs  

\EventEditors{John Q. Open and Joan R. Access}
\EventNoEds{2}
\EventLongTitle{42nd Conference on Very Important Topics (CVIT 2016)}
\EventShortTitle{CVIT 2016}
\EventAcronym{CVIT}
\EventYear{2016}
\EventDate{December 24--27, 2016}
\EventLocation{Little Whinging, United Kingdom}
\EventLogo{}
\SeriesVolume{42}
\ArticleNo{23}

\newif\iflongversion
\newif\ifshortversion

\longversiontrue
 \shortversionfalse 

\begin{document}

\maketitle

\begin{abstract}
  An important area of combinatorial optimization is the study of packing and covering problems, such as \acl{BP}, \acl{MKP}, and \acl{BC}.
  Those problems have been studied extensively from the viewpoint of approximation algorithms, but their parameterized complexity has only been investigated barely. 
  For problem instances containing no ``small'' items, classical matching algorithms yield optimal solutions in polynomial time.
  In this paper we approach them by their \emph{distance from triviality}, measuring the problem complexity by the number $k$ of small items.
  
  Our main results are fixed-parameter algorithms for vector versions of \acl{BP}, \acl{MKP}, and \acl{BC} parameterized by $k$.
  The algorithms are randomized with one-sided error and run in time $4^k\cdot k!\cdot n^{O(1)}$.
  To achieve this, we introduce a colored matching problem to which we reduce all these packing problems.
  The colored matching problem is natural in itself and we expect it to be useful for other applications.
  We also present a deterministic fixed-parameter algorithm for \acl{BP} with run time $\cO((k!)^2 \cdot k \cdot 2^k \cdot n \log(n))$.
\end{abstract}

\section{Introduction}
\label{sec:introduction}
An important area of combinatorial optimization is the study of packing and covering problems.
Central among those is the {\sc Bin Packing} problem, which has sparked numerous important algorithmic techniques.
In {\sc Bin Packing}, the goal is to pack a set of $n$ items with sizes in $(0,1]$ into as few unit-sized bins as possible.
Referring to its simplicity and vexing intractability, this problem has been labeled as ``the problem that wouldn't go away'' more than three decades ago~\cite{GareyJ1981} and is still the focus of groundbreaking research today.
Regarding approximability, the best known is an additive $O(\log \textnormal{OPT})$-approximation, due to Hoberg and Rothvo{\ss}~\cite{GoemansR2014,HobergR2017}.

A recent trend is to apply tools from parameterized complexity theory to problems from operations research~\cite{MnichvB2018}.
For {\sc Bin Packing}, a natural parameter is the minimum number of bins.
For this parameter, Jansen et al.~\cite{JansenKMS2013} showed that this problem is $\mathsf{W}[1]$-hard, even for instances encoded in unary.
Another natural parameter is the number~$d$ of distinct item sizes.
For $d = 2$, a polynomial-time algorithm was discovered by McCormick et al.~\cite{McCormickSS1997,McCormickSS2001} in the 1990s. 
The complexity for all $d\geq 3$ was open for more than 15 years, until a breakthrough result of Goemans and Rothvo{\ss}~\cite{GoemansR2014} showed that {\sc Bin Packing} can be solved in time $(\log\Delta)^{2^{2^d}}$, where $\Delta$ is the largest number in the input.
A similar result was shown later by Jansen and Klein~\cite{JansenK2017}.
Neither the algorithm by Goemans and Rothvo{\ss} nor the algorithm by Jansen and
Klein are fixed-parameter algorithms for parameter $d$, which would require the
algorithm to run in time $f(d)\cdot n^{O(1)}$ for some computable function
$f$\footnote{Although the algorithm by Jansen and Klein is a fixed-parameter
  algorithm for the parameter $|V_{I}|$~--~the number of vertices of the
  integer hull of the underlying knapsack polytope.}.

In light of these daunting results, we propose another natural parameter for {\sc Bin Packing}.
This parameter is motivated by the classical approach of parameters measuring the \emph{distance from triviality}---a concept that was first proposed by Niedermeier~\cite[Sect. 5.4]{Niedermeier2006}.
Roughly speaking, this approach measures the distance of the given instance from an instance which is solvable in polynomial time. 
This approach was already used for many different problems such as \textsc{Clique}, \textsc{Set Cover}, \textsc{Power Dominating Set}, or \textsc{Longest Common Subsequence}~\cite{GuoHN2004}.
Even one of the arguably most important graph parameters~---treewidth---is often interpreted as the distance of a given
graph from a tree~\cite{GuoHN2004}.
Interestingly, the number of special cases where \acl{BP} can be solved in polynomial time is rather small and the corresponding algorithms
often rely on reductions to matching problems. 
In this work, we propose as novel parameter the distance from instances without \emph{small items}. 
If no small item (with size at most $1/3$) exists, \acl{BP} becomes polynomial-time solvable via a reduction to the matching problem as each bin can contain at most two items.
If the number of small items is unbounded, the problem becomes $\Class{NP}$-hard.

Two related problems to \acl{BP} are \acl{BC}, where the number of covered bins
(containing items of total size at least $1$) should be maximized, and
\acl{MKP}---a generalization of the {\sc Knapsack} problem. These problems have
been studied extensively (see the books by Gonzalez~\cite{Gonzalez2007} and Kellerer et al.~\cite{KellererPP2004}).
They share the \acl{BP} trait that the efficiency of exact algorithms is hindered by the existence of small objects.

In all mentioned problems, the items have a one-dimensional size requirement.
As this is too restrictive in many applications, so-called \emph{vector versions} were proposed~\cite{AlonACESVW1998,GareyGJY1976}. 
In these versions, called \acl{VP}, \acl{VC}, and \acl{VMKP}, each object has a $d$-dimensional size requirement and a set of objects can be packed only if the size constraints are fulfilled in \emph{each} dimension $j=1,\dots,d$.
These problems are much harder than their $1$-dimensional version, e.g., \acl{VP} does not admit an \acl{APTAS} even for $d=2$~\cite{Woeginger1997}.
For $d$-dimensional problems, we use the word \emph{vectors} instead of \emph{items} and \emph{containers} instead of \emph{bins}.

\vspace{-1em}
\subparagraph*{What it means to be small.}
In the one-dimensional version of \acl{VP}, the definition of a small item is quite natural: Every item with size less or equal than $1/3$ is considered small.
As a consequence, each bin can contain at most two large items.
We would like to transfer this property from one dimension to the $d$-dimensional case. 

The requirement for large items is that only two of them can be placed inside the same container.
We call a subset of vectors $\vectors' \subseteq\vectors$ \emph{3-incompatible} if no selection of three distinct vectors from
$\vectors'$ may be placed in the same container, i.e., for each $u,v,w\in \vectors'$ there exists an $\ell \in \{1, \dots, d\}$
such that $u^{\ell} + v^{\ell} + w^{\ell} > T^{\ell}$, where $T^{\ell}$ is the \emph{capacity constraint} of the container in dimension~$\ell$.
Let $\vectors_L \subseteq\vectors$ be a \emph{largest} 3-incompatible set; we call the vectors $v \in \vectors_L$ \emph{large} and call the vectors from the set $\vectors_{S} = \vectors \setminus \vectors_L$ \emph{small}.
Moreover, we define the \emph{number of small vectors} in $\vectors$ as the cardinality of the complement of a largest 3-incompatible set in~$\vectors$.
Note that each 3-incompatible set $\vectors'$ contains at most two vectors where all the entries have size of at most $1/3$.
Hence, for \acl{BP} the largest 3-incompatible set corresponds to the set of large items plus at most two additional items. 


An important property of our definition is that the smallness of a vector is no longer an attribute of the vector itself, but needs to be treated with regard to all other vectors. 
Finding a set $\vectors_{S}\subseteq \vectors$ of small vectors of minimum
cardinality might be non-trivial. 
We argue that this task is fixed-parameter tractable parameterized by $|\vectors_{S}|$.
To find $\vectors_{S}$, we compute the largest 3-incompatible set in $\vectors$.
The complement $\vectors_{S} = \vectors \setminus \vectors_L$ of a largest 3-incompatible set can be found in time $f(|\vectors_{S}|)\cdot n^{O(1)}$ by a reduction to \acl{3HS}.
In this problem, a collection of sets $S_1, \dots, S_n \subseteq T$ with $|S_i| = 3$ is given, and a set~$H \subseteq T$ with $H\cap S_i$ for all $i \in \{1,\dots, n\}$ is sought.
In \Cref{sec:ReductionToOverTheRainbow}, we present a reduction from the problem of finding the sets $\vectors_{L}$ and $\vectors_{S}$ to an instance of the \acl{p3HS} problem, which we can solve using:

\begin{fact}[\cite{FafianieK2015,NiedermeierR2003,vanBevern2014}]
\label{t:HittingSet}
  \acl{p3HS} can be solved in time $2.27^k \cdot n^{O(1)}$, where $k$ is the size of the solution.
  A corresponding solution can obtained within the same time.
\end{fact}
%
%

\subparagraph*{Our results.}
We settle the parameterized complexity of the vector versions of \acl{BP}, \acl{BC}, and \acl{MKP} parameterized by the number $k$ of small objects.
Our main results are randomized fixed-parameter algorithms, which solve all those problems in time $O(k!)\cdot n^{O(1)}$ with one-sided error where $n$ is the total number of objects.
Note that \acl{VMKP} is already $\Class{NP}$-hard for $d=1$ and $p_{\max}\leq n^{O(1)}$~\cite{Geng1986,MartelloT1990} where $p_{\max}$ denotes the largest profit of any object.
\begin{theorem}\label{thm:fpt-packing-algorithms}
  {\sc Vector Packing} and {\sc Vector Covering} can be solved by a randomized algorithm (with bounded false negative rate in $n$) in  time $4^k\cdot k!\cdot n^{O(1)}$.
  {\sc Vector Multiple Knapsack} can be solved by a randomized algorithm (with bounded false negative rate in $n+p_{\max}$) in time $4^k\cdot k!\cdot n^{O(1)}\cdot (p_{\max})^{O(1)}$ where $p_{\max}$ is the largest profit of any vector. 
\end{theorem}
Our approach is to reduce the vector versions of packing and covering problems
to a new matching problem on edge-colored graphs, which we call \acl{POTRM}.

In the \acl{POTRM} problem, we are given a graph $G$. Each edge $e\in E(G)$ is
assigned a set of colors $\lambda(e)\subseteq \colors$ and for each color, there
is a non-negative weight $\gamma(e,c)$.
The objective is to find a perfect matching $M$ of $G$ and a function $\xi\colon
M\to \colors$ such that (i) $\chi(e)\in \lambda(e)$ for all $e\in M$ (we
can only choose from the assigned colors), (ii) $\bigcup_{e\in
  M}\chi(e) = \colors$ (every color is present in the matching), and (iii)
$\sum_{e\in M}\gamma(e,\chi(e))$ is minimized (the sum of the weights is
minimized).
The parameter for the problem is $|\colors|$, the number of different colors.

We show how to solve \acl{POTRM} by an approach that is based on the \acl{CM} problem.
The \acl{CM} problem was proposed by Sorge et al.~\cite{SorgeBNW2012}, who asked whether it is fixed-parameter tractable.
The question was resolved independently by Gutin et al.~\cite{GutinWY2017}, and by Marx and Pilipczuk~\cite{MarxP2014}, who both gave randomized fixed-parameter algorithms.
Based on both results, we derive:
\def\claimedtheorem{There is a randomized algorithm (with bounded false negative rate in $n+\ell$) that solves {\sc Perfect Over-The-Rainbow Matching} in time $2^{|C|}\cdot n^{O(1)}\cdot \ell^{O(1)}$.}
\begin{theorem}
\label{thm:over-the-rainbow-algo}
  \claimedtheorem
\end{theorem}
This algorithm forms the backbone of our algorithms for \acl{VP}, \acl{VC}, and \acl{VMKP}.

Whether there is a deterministic fixed-parameter algorithm for \acl{CM} remains a challenging question, as also pointed out by Marx and Pilipczuk~\cite{MarxP2014}.
For some of the problems that can be solved by the randomized algebraic techniques of Mulmeley, Vazirani and Vazirani~\cite{MulmuleyVV1987}, no deterministic polynomial-time algorithms have been found, despite significant efforts.
The question is whether the use of such matching algorithms is essential for \acl{CM}, or can be avoided by a different approach.

We succeed in circumventing the randomness of our algorithm in the $1$-dimensional case of \acl{pBP}.
Namely, we develop another, deterministic algorithm for {\sc Bin Packing}, for which we prove strong structural properties of an optimal solution; those structural insights may be of independent interest.
\begin{theorem}
  {\sc Bin Packing} can be solved deterministically in  time $\cO((k!)^2 \cdot k \cdot 2^k \cdot n  \log(n)) $. 
\end{theorem}

\subparagraph*{Related Work.}
The class of small items, their relation to matching problems, and special instances without small items have been extensively studied in the literature: Shor~\cite{Shor1985,Shor1986} studies the relation between \emph{online bin packing,} where the items are uniformly
randomly chosen from $(0,1]$, and the matching problem on planar graphs.
Those problems are closely related as almost all bins in an optimal solution contain at most two items.
Csirik et al.~\cite{CsirikFLZ1990} study the Generalized First-Fit-Decreasing
heuristic for \acl{VP} and show that their strategy is optimal for instances
that contain at most two small items.
Kenyon~\cite{Kenyon1996} studies the expected performance ratio of the \emph{Best-Fit} algorithm for \acl{BP} on a worst-case instance where the items arrive in random order.
To prove an upper bound on the performance ratio, she classifies items into small items (size at most~$1/3$), medium items
(size at least $1/3$ and at most $2/3$), and large items (size at least $2/3$)~\cite{Kenyon1996}.
Kuipers~\cite{Kuipers1998} studies so-called \emph{bin packing games} where the goal is to share a certain profit in a fair way between the players controlling the bins and players controlling the items.
He only studies instances without small items and shows that every such instances has a non-empty $\varepsilon$-core (a way of spreading the profits relatively fair) for $\varepsilon \geq 1/7$.
Babel et al.~\cite{BabelCKK2004} present an algorithm with competitive ratio $1+1/\sqrt{5}$ for online bin packing without small items.
In another online version of the problem, the items and a conflict graph on them are given offline and, then online, variable-sized bins arrive.
The case that the conflict graph is the union of two cliques corresponds to instances with no small items and was studied by Epstein et al.~\cite{EpsteinFL2011}.
Another version of {\sc Bin Packing} forbids to pack more then $k$ different items into a single bin.
The special case $k=2$ corresponds to instances without small items and can be solved in time $n^{O(1/\varepsilon^{2})}$ for bins of
size $1+\varepsilon$~\cite{EpsteinLS2012}.
Finally, Bansal et al.~\cite{BansalEK2016} study approximation algorithms for \acl{VP}.
To obtain their algorithms, they present a structural lemma that states that any solution with $m$ bins can be turned into a solution with $(d+1)m/2$ bins such that each bin either contains at most two items or has empty space left in all but one dimensions.
This result is then used to reduce the problem to a multi-objective budgeted matching problem.
  
From an approximation point of view, the problems considered in this work have been studied extensively, both for the $1$-dimensional variant as well as for the vector versions. 
We refer to the survey of Christensen et al.~\cite{ChristensenKPT2017} for an overview. 
Regarding parameterized algorithms for problems from operations research, the resulting body of literature is too large for a detailed description and we refer to the survey of Mnich and van Bevern~\cite{MnichvB2018}.
Some of the $1$-dimensional variants of the problems considered in this work
have been studied from a parameterized
perspective~\cite{JansenK2017,JansenKMS2013, JansenS2011}. 
In contrast, to the best of our knowledge, there are no such results for the vector versions of these problems.


\iflongversion
\subparagraph*{Structure of the document.}
In \Cref{sec:prelim}, we briefly introduce randomized and parameterized algorithms and we define the problems studied in this work with their corresponding parameters.
\iflongversion We use \Cref{sec:ReductionToOverTheRainbow} to show the parameterized reductions from the packing and covering problems to the \acl{POTRM} problem.
\else
We present a parameterized reduction from \acl{pVC} to \acl{pPOTRM} in \Cref{sec:ReductionToOverTheRainbow}.
Similar reductions from the other packing and covering problems to \acl{pPOTRM} can be found in the full version.
\fi \Cref{sec:solve_otrpm} contains the parameter preserving transformation of the \acl{POTRM} problem to \acl{CM} and the resulting parameterized algorithms.
As the algorithm for \acl{POTRM} is randomized, so are the algorithms for the covering and packing problems.
In \Cref{sec:deterministic}, we give a deterministic parameterized algorithm for the classical $1$-dimensional version of
\acl{BP}.
Lastly, we summarize our results and state some open questions for further possible research. 
\fi

\section{Preliminaries: Parameterized and Randomized Algorithms}
\label{sec:prelim}
We give a short introduction to parameterized and randomized algorithms, and refer to the standard textbooks for
details~\cite{CyganFKLMPPS2015,Niedermeier2006}.
Afterwards, we introduce the packing problems formally.
Finally, we define our auxiliary matching problem.

\subparagraph*{Parameterized Algorithms.}
A \emph{parameterized problem} is a language~$L\subseteq \{0,1\}^{*}\times \mathbb{N}$ where the second element is called the \emph{parameter}.
Such a problem is \emph{fixed-parameter tractable} if there is an algorithm that decides whether $(x,k)$ is in $L$ in time $f(k)\cdot |x|^c$ for a computable function $f$ and constant $c$.
A \emph{parameterized reduction} from a parameterized problem~$L$ to another one~$L'$ is an algorithm that transforms an
instance~$(x,k)$ into~$(x',k')$ such that (i)~$(x,k)\in L\Leftrightarrow(x',k')\in L'$, (ii) $k'\leq f(k)$, and (iii) runs in time $f(k)\cdot |x|^c$. 

\subparagraph*{Randomized Algorithms.}
A \emph{randomized algorithm} is an algorithm that explores some of its computational paths only with a certain probability.
A randomized algorithm~$\mathsf{A}$ for a decision problem~$L$ has \emph{one-sided error} if it either correctly detects positive or negative instances with probability $1$.
It has a \emph{bounded false negative rate} if $\Pr[\mathsf{A}(x) = \text{``no''}\mid x\in L]\leq 1/|x|^c$, that is, it declares a ``yes''-instance as a ``no''-instance with probability at most $1/|x|^{c}$.
All randomized algorithms in this article have bounded false negative rate.

\subparagraph*{Packing and Covering Problems.}
In the {\sc Vector Packing} problem we aim to pack a set $\vectors = \{v_1, \dots, v_n\}\subseteq \mathbb{Q}^d_{\geq 0}$ of vectors into the smallest possible number of containers, where all containers have a common capacity constraint $T\in \mathbb{Q}^d_{\geq 0}$.
Let $v_j\in\mathcal{V}$ be a vector. We use $v^\ell_j$ to denote the $\ell^{\textnormal{th}}$ component of $v_j$ and $T^\ell$ to denote the $\ell^{\textnormal{th}}$ constraint.
A \emph{packing} is a mapping $\sigma \colon \vectors \to \NN_{>0}$ from vectors to containers.
It is \emph{feasible} if all containers~$i \in \NN_{>0}$ meet the capacity constraint, that is, for each $\ell \in \{1,\dots,d\}$ it holds that $\sum_{v_j \in \sigma^{-1}(i)} v^\ell_j \leq T^\ell$.
Using as few containers as possible means to minimize $\max\{\sigma(v_{j})\,|\, v_j \in \vectors\}$.


In the introduction we already discussed what it means to be ``small''.
We expect only few small items, so we consider this quantity as parameter for \acl{VP}:
\begin{center}
  \framebox[1.0\textwidth]{
    \begin{tabularx}{0.97\textwidth}{rXl}
      \multicolumn{3}{X}{{\sc\centering \acl{pVP}} \hfill  \textit{Parameter:} Number $k$ of small vectors} \\[.5ex]
      \textit{Input:}      & \multicolumn{2}{l}{A set $\vectors =
        \{v_1, \dots, v_n\}\subseteq \mathbb{Q}^d_{\geq 0}$ vectors and capacity constraints $T \in \mathbb{Q}^d_{\geq 0}$.}\\
      \textit{Task:}	&  \multicolumn{2}{l}{Find a packing of $\vectors$ into the smallest number of containers.}\\
  \end{tabularx}}
\end{center}

\noindent The 1-dimensional case of the problem is the \acl{BP} problem.
There, vectors are called \emph{items}, their single component \emph{size} and the containers \emph{bins}. 
In contrast to the multi-dimensional case, we are now given a \emph{sequence} of items, denoted as $\items$.\footnote{This is due to the fact that in the   multi-dimensional setting, we can simply model multiple occurrences of the same vector by introducing an additional dimension encoding the index of the vector. 
This is not possible in the one-dimensional case.}
\iflongversion

\begin{center}
  \framebox[1.0\textwidth]{
    \begin{tabularx}{0.97\textwidth}{rXl}
      \multicolumn{3}{X}{{\sc\centering \acl{pBP}} \hfill  \textit{Parameter:} Number $k$ of small items} \\[.5ex]
      \textit{Input:}      & \multicolumn{2}{l}{A sequence $\items =
        (i_1, \dots, i_n)$ of $n$ items such that $i_j \in
        \mathbb{Q}^1_{\geq 0}$ for each $i_j \in \items$,}\\
      & \multicolumn{2}{l}{and a capacity constraint $T \in \mathbb{Q}^1_{\geq 0}$.}\\
      \textit{Task:}	&  \multicolumn{2}{l}{Find a packing of $\items$ into the smallest number of bins.}\\
  \end{tabularx}}
\end{center}
\fi
Another related problem is \acl{VC}, where we aim to \emph{cover} the containers.
We say a packing $\sigma \colon \vectors \to \NN_{>0}$ covers a container~$i$ if ${\sum\nolimits_{v_j \in \sigma^{-1}(i)} v_j^\ell } \geq T^\ell$ for each component $\ell \in \{1,  \dots, d\}$.
The objective is to find a packing $\sigma$ that maximizes the number of covered containers, that is, we want to maximize $|\{ i \in \NN \,|\, \sum\nolimits_{v \in \sigma^{-1}(i)} v_j^\ell  \geq T^\ell \, \text{ for all } \ell \in \{1,  \dots, d\} \}|$. 
\iflongversion
\begin{center}
  \framebox[1.0\textwidth]{
    \begin{tabularx}{0.97\textwidth}{rXl}
      \multicolumn{3}{X}{{\sc\centering \acl{pVC}} \hfill  \textit{Parameter:} Number $k$ of small vectors} \\[.5ex]
      \textit{Input:}      & \multicolumn{2}{l}{A set $\vectors = \{v_1, \dots v_n\}$ of $n$ vectors of dimension $d$ such that $v_j \in \mathbb{Q}^d_{\geq 0}$ for each }\\
      & \multicolumn{2}{l}{$v_j \in \vectors$, as well as capacity constraints $T \in \mathbb{Q}^d_{\geq 0}$.}\\
      \textit{Task:}	&  \multicolumn{2}{l}{Find a covering of $\vectors$ into the largest number of containers.}\\
  \end{tabularx}}
\end{center}
\fi
%
The last problem we study is the \acl{VMKP} problem: Here a packing into a finite number of $C$ many containers is sought.
Therefore, not all vectors may fit into them.
We have to choose which vectors we pack considering that each vector $v_j \in \vectors$ has an associated profit $\profit{v_j} \in \NN_{\geq 0}$.
A packing of the vectors is a mapping $\sigma \colon \vectors \to \{1, \dots, C\} \cup \{\bot\}$   such that  $\sum_{v_j \in \sigma^{-1}(i)} v_j^\ell \leq T^\ell$ holds for all $i \in \{1, \dots, C\}$ and $\ell \in \{1, \dots, d\}$, which means no container is over-packed.
The objective is to find a packing with a maximum total profit of the packed items, that is, we want to maximize $\sum_{i = 1}^C\sum_{v \in \sigma^{-1}(i)} \profit{v}$.
\iflongversion
\begin{center}
  \framebox[1.0\textwidth]{
    \begin{tabularx}{0.97\textwidth}{rXl}
      \multicolumn{2}{X}{{\sc\centering \acl{pVMKP}} \hfill  \textit{Parameter:} Number $k$ of small vectors \phantom{xx}} \\[.5ex]
      \textit{Input:}      & \multicolumn{2}{l}{A set $\vectors =
        \{v_1, \dots v_n\}$ with $v_i \in \mathbb{Q}^d_{\geq 0}$, a profit function $p\colon\vectors \to
        \NN_{\geq 0}$,}\\
      & \multicolumn{2}{l}{as well as capacity constraints $T \in \mathbb{Q}^d_{\geq 0}$ and a number of bins $C$.}\\
      \textit{Task:}	&  \multicolumn{2}{l}{Find a packing of $\vectors$ into the bins which maximizes the profit.}\\
  \end{tabularx}}
\end{center}
\fi


\subparagraph*{Conjoining and Over-the-Rainbow Matchings.}
We introduce two useful problems to tackle the questions mentioned above, namely \acl{pPOTRM} and \acl{pCM}.
\iflongversion
The following section presents the reductions from \acl{pVP}, \acl{pVC} and \acl{pVMKP} to \acl{pCM} using \acl{pPOTRM} as an
intermediate step.
By results of Gutin et al.~\cite{GutinWY2017}, and Marx and Pilipczuk~\cite{MarxP2014}, we can solve \acl{pCM} efficiently and, thus, our packing and covering problems as well.
\fi
A \emph{matching} in a graph $G$ describes a set of edges $M \subseteq E(G)$ without common nodes, that is, $e_1\cap e_2 =
\emptyset$ for all distinct $e_1, e_2 \in M$.
A matching is \emph{perfect} if it covers all nodes.
In the \acl{pPOTRM} problem, we are given an graph $G$ as well as
a \emph{color function} $\lambda \colon E(G) \rightarrow 2^{\colors}\setminus
\{\emptyset\}$ which assigns a non-empty set of colors to each edge, and an
integer $\ell$. For each edge $e$ and each color $c\in \lambda(e)$, there is a
non-negative weight $\gamma(e,c)$. 
The objective is to find a perfect matching $M$ and a surjective function $\xi
\colon M \rightarrow \colors$ with $\xi(e) \in \lambda(e)$ for each $e \in M$
such that $\sum_{e\in M}\gamma(e,\xi(e))\leq \ell$. 
The surjectivity guarantees that each color must appear at least once. 
We call such a pair $(M,\xi)$ a perfect \textit{over-the-rainbow matching} and
the term $\sum_{e\in
  M}\gamma(e,\xi(e))$ denotes its \emph{weight}. 
This name comes from the closely related \textit{rainbow matching} problem, where each color appears exactly once~\cite{KanoL2008,LeP2014}.
In contrast to our problem, a sought rainbow matching  covers as many colors as possible, but not necessarily all,  and the maximum size of a rainbow matching is bounded by the number of colors. In our variant we must cover all colors, and likely have to cover some colors more than once to get a perfect matching.
Formally, the problem is defined as follows:
\begin{center}
  \framebox[1.0\textwidth]{
    \begin{tabularx}{0.97\textwidth}{rXl}
      \multicolumn{2}{X}{{\sc\centering \acl{pPOTRM}} \hfill \textit{Parameter:} The number of colors $|\mathcal{C}|$ \phantom{x}} \\[.5ex]
      \textit{Input:}      & \multicolumn{2}{l}{A graph $G$, a set of colors $\colors = \{1, \dots, |\mathcal{C}|\}$, a function $\lambda \colon E \rightarrow 2^{\colors}\setminus \{\emptyset\}$,}\\
      & \multicolumn{2}{l}{edge weights $\gamma: \{(e,c)\mid e \in E(G), c \in \lambda(e)\} \to \mathbb{Q}_{\geq 0}$, and a number $\ell$}\\
      \textit{Task:}	&  \multicolumn{2}{l}{Find a perfect over-the-rainbow matching $(M,\xi)$ in $G$ of weight at most $\ell$.}\\
  \end{tabularx}}
\end{center}
We sometimes omit the surjective function $\xi$, if it is clear from the
context. 

Related to this problem is \acl{pCM}: We have a partition $V_1 \uplus \dots \uplus V_t$ of the nodes of $G$ and a pattern graph $H$ with $V(H)=\{V_1,\dots,V_t\}$.
Instead of covering all colors in a perfect matching, this problems asks to find a \emph{conjoining} matching~$M\subseteq E(G)$, which is a perfect matching such that for each $\{V_i,V_j\} \in E(H)$ there is an edge in $M$ with one node in $V_i$ and the other in $V_j$.
Roughly speaking, each edge in $H$ corresponds to some edges in $G$ of which at least one has to be taken by $M$. 
%
Formally, the problem is given by:
\begin{center}
  \framebox[1.0\textwidth]{
    \begin{tabularx}{0.97\textwidth}{rXl}
      \multicolumn{2}{X}{{\sc\centering \acl{pCM}} \hfill  \textit{Parameter:} The number of edges of $H$\phantom{x}} \\[.5ex]
      \textit{Input:}      & \multicolumn{2}{l}{A weighted graph $G = (V, E,\gamma)$ with $\gamma\colon E\rightarrow
\QQ_{\geq 0}$, a node partition $V_1 \uplus \dots \uplus V_t$,}\\
      & \multicolumn{2}{l}{a number $\ell$, and a graph $H$ with $V(H)=\{V_1,\dots,V_t\}$}\\
      \textit{Task:}	&  \multicolumn{2}{l}{Find a perfect matching $M$ in $G$ of weight at most $\ell$ such that}\\
      &  \multicolumn{2}{l}{for each edge $\{V_i,V_j\} \in E(H)$ there is an edge $\{u,v\}\in M$ with $u\in V_i$ and $v\in V_j$.}\\
  \end{tabularx}}
\end{center}
Gutin et al.~\cite[Theorem 7]{GutinWY2017} and Marx and Pilipczuk~\cite{MarxP2014} gave randomized fixed-parameter algorithms for \acl{pCM} on loop-free graphs $H$.
We show how a simple reduction also solves the problem on graphs with loops.
\begin{lemma}
\label{t:GutinCGM}
  The {\sc Conjoining Matching} problem can be solved by a randomized algorithm (with bounded false negative rate in $n+\ell$) in time
  \smash{$2^{|E(H)|}\cdot n^{O(1)}\cdot \ell^{O(1)}$}, even if~$H$ contains self-loops.
\end{lemma}
\begin{proof}[Sketch of Proof.]
  If $H$ does not contain self-loops the claim is proven by Gutin et al.~\cite{GutinWY2017}.
  The case that $H$ does contain self-loops can be reduced to the loop-free version by a simple layering argument: First direct the edges of $H$ arbitrarily (for instance by using the lexicographical order of the nodes) and then define $G'$ and $H'$ as
  \begin{align*}
    V(H') &= \{\,h',h''\mid h\in V(H)\,\}\cup\{\,h^*\,\},\\
    E(H') &= \{\,\{h_i',h_j''\}\mid \{h_i,h_j\}\in E(H) \,\},\\
    V(G') &= \{\,v',v'',v^*\mid v\in V(G)\,\},\\
    E(G') &= \{\, \{v',v^*\},\{v'',v^*\}\mid v\in V(G)\,\}
            \cup\{\,\{v',w''\}\mid \{v,w\}\in E(G)\,\}.
  \end{align*}
  Observe that $H'$ is loop-free, and $|E(H)|=|E(H')|$.
  Further note that, in any perfect matching in $G'$, for each $v\in V(G)$ either $v'$ or $v''$ must be matched with $v^{*}$; the other node together with its matching partner corresponds to an edge in a corresponding perfect matching in $G$ as it is only connected to $v^*$ or $\{w',w''\mid \{v,w\}\in E(G)\}$.
  Finally, to preserve weights, set $\gamma'(\{v',v^*\})=\gamma'(\{v'',v^*\})=0$ and  $\gamma'(\{v',w''\}=\gamma(\{v,w\})$ for all $v,w\in V(G)$.
\end{proof}
      
\section{Reducing Packing and Covering Problems to Finding Perfect Over-the-Rainbow Matchings}
\label{sec:ReductionToOverTheRainbow}
\iflongversion
In this section, we prove \Cref{thm:fpt-packing-algorithms}, i.\,e., we show that \acl{pVP} and \acl{pVC} can be solved in time $4^k\cdot k!\cdot n^{O(1)}$ and the \acl{pMKP} problem in time $4^k\cdot k!\cdot n^{O(1)} \cdot (p_{\max})^{O(1)}$ respectively.
\fi
The first phase to solve these packing and covering problems is to interpret them as \acl{pPOTRM} problems.
Each problem admits a similar procedure: Guess the packing of the small vectors; guess the number of large vectors for each container; use these guesses to pack the large vectors by formulating the problem as a matching problem in a graph.
The idea is that the nodes of this graph represent the large vectors.
An edge represents that both endpoints can be placed into the same container to satisfy the condition of the problem, i.\,e., either to fit into the container or to cover it.
Introducing a color and a weight function for the edges, we manage to handle the containers already filled with some small vectors and the overall profits of the packing.
Note that the guessing also serves as a transformation from the minimization and maximization problems to decision problems as each guess also corresponds to some fixed number of containers and if applicable to the profit. So we ask if there is a solution with these numbers and thus we can solve this question via a reduction.

\subparagraph{Identifying the Set of Small Vectors}
Before we can proceed as mentioned above, we first need to identify the sets $\vectors_L$ and $\vectors_S$ of large and small vectors explicitly.
This can be done via a reduction to the \acl{3HS} problem as follows: The set of elements is given by the set of vectors $\vectors$ and we compute all sets $S\subseteq \vectors$ of triplets that fit together in a single container, i.\,e., $|S| = 3$ and $\sum_{v\in S} v \leq T$.
Consider a hitting set $H$ for this instance.
Then the set $\vectors\setminus H$ is large.
To see this, consider any three distinct vectors $u,v,w\in \vectors\setminus H$.
If we had $u+v+w \leq T$, then the set $\{u,v,w\}$ would be part of the computed selection of subsets.
Yet, $\{u,v,w\}\cap H = \emptyset$---a contradiction.
We can pick the given number of small vectors~$k$ and use \cref{t:HittingSet} to obtain a hitting set $H\subseteq\vectors$ of size at most $k$.
We set $\vectors_L = \vectors\setminus H$ and $\vectors_S = H$.
As there are $\cO(n^3)$ sets of triplets this yields a run time of $2.27^k \cdot n^{O(1)}$ (see \cref{t:HittingSet}).

\subparagraph*{The Case of Packing Vectors}
Recall that in the \acl{pVP} problem we are given $n$ vectors of dimension $d$ and a set of containers, each with the same size
limitation $T \in \QQ^d$.
Furthermore, we assume that the sets $\vectors_S$ and $\vectors_L$ are given explicitly by using the computation explained above.
Any solution needs at most $|\vectors|$ and at least $\lceil|\vectors_L|/2\rceil$ containers.
Furthermore, if there is a solution with $m\leq |\vectors|$ there is also a solution with $m'$ containers for any $m'\in\{m+1,\dots,|\vectors|\}$.
Thus a binary search for the optimal number of containers between the given bounds is possible.
Let $C$ be the current guess of the number of containers.
Now we have to decide whether there exists a solution using exactly $C$ containers.

We guess the packing of the small vectors, that is, we try all possible partitions into at most $\min\{C,k\}$ subsets. 
It is not hard to see that the number of such partitions is upper bounded by the $k^{\textnormal{th}}$ Bell number:
The first vector is packed by itself, the second can either be packed with the first one or also by itself, and so on.
If any of the corresponding containers is already over-packed, we discard the guess.
In the following, we call the used containers \emph{partially filled} as some area is already occupied by small vectors. 
For these partially filled containers, we guess which of them are finalized, i.\,e., which of them do not contain an additional large vector in the optimal solution, and discard them for the following steps.
There are at most~$2^k$ such guesses.
We denote the number of discarded containers as $C_0$.
For each of the remaining partially filled containers, we introduce a new color. 
Furthermore, we introduce a color $\top$ representing the empty containers if existent. 
Hence, the resulting set of colors $\colors$ has a cardinality of at most $k+1$.
For each $c \in \colors$, we denote by $s(c) \in \QQ^d$ the residual size in the corresponding container.

We place the large vectors $\vectors_L$ inside the $C-C_0$ residual containers by reducing it to a \acl{pPOTRM} problem.
Note that if the current guesses are correct, each of the $C-C_0$ containers receives at least one and at most two large vectors.
Hence, we may assume $|\vectors_L|/2 \leq (C-C_0) \leq |\vectors_L|$ (and reject the current guess otherwise).
Furthermore, the number of containers receiving one or two large items, respectively, is already determined by~$C$ and $C_0$.
We denote these numbers by $C_1$ and $C_2$ and remark that $C_2 = |\vectors_L| - (C-C_0)\geq 0$ and $C_1 := (C-C_0) - C_2 = 2(C-C_0) - |\vectors_L| \geq 0$.

We now construct a graph $G = (V, E)$ to find a feasible packing.
Every large vector $v\in \vectors_{L}$ is represented by two nodes $v$ and $v'$ in $V$.
Let $\vectors'_{L}=\{v'\mid v\in \vectors_{L}\}$.
Next, we define a set $\mathcal{B}$ of $2\cdot C_2$ new nodes called \emph{blocker nodes}, which ensures that all vectors are placed inside exactly $(C-C_0)$ containers.
We define $V := \vectors_L \cup \vectors_L' \cup \mathcal{B}$. In this graph, an edge between the nodes in $\vectors_L \cup \vectors_L'$ represents a possible packing of the large vectors inside one container.
Hence, we add an edge $e = \{v,w\}$ between two original vectors $v,w \in \vectors_L$ and assign this edge some color $c \in \colors$ if these vectors fit together inside the corresponding container.
Furthermore, we add an edge between a vector $v \in \vectors_L$ and its copy $v' \in \vectors_L'$ and assign it the color $c \in \colors$ if the vector alone fits inside the corresponding container.
More formally, we introduce the set of edges $E_c := \{\{u, v\}\, | \, u,v \in \vectors_L, u + v \leq s(c)\} \cup \{\{v,v'\}\,|\, v \in \vectors_L, v \leq s(c)\}$ for each color $c \in \colors$.
Additionally, we introduce the edges of a complete bipartite graph between the copied nodes~$\vectors_L'$ on the one hand and the blocker nodes $\mathcal{B}$ on the other hand.
More formally, we define $E_{\bot} := \{\{v',b\} \,|\, v' \in \vectors_L', b \in \mathcal{B}\}$.
Together, we get $E := E_{\bot} \cup \bigcup_{c \in \colors} E_c$.
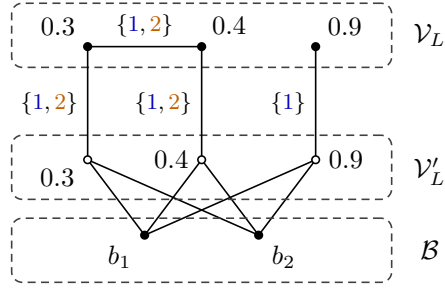
\begin{figure}
  \centering    
  \def\firstColor#1{\textcolor{blue!75!black}{#1}}
  \def\secondColor#1{\textcolor{orange!75!black}{#1}}
  \begin{tikzpicture}[scale=0.50,
    dot/.style={circle,semithick, draw=black,fill=black,inner sep=1pt,outer sep=0pt},
    every label/.style={label distance=0.01mm},
    mybox/.style = {
        draw, semithick,
        rounded corners, densely dashed,
        inner sep=0pt, minimum width=5cm,
        minimum height=0.85 cm, color=darkgray
      }
    ]   
    \node[dot, label={145:$0{.}3$}] (p3) at (0,0)  {};
    \node[dot, label={75:$0{.}4$}] (p4) at (3cm,0) {};
    \node[dot, label={45:$0{.}9$}] (p9) at (6cm,0) {};
    \node (p) at (9cm,3mm) {$\vectors_{L}$};
    \node[dot, fill=white, label={185:$0{.}3$}] (q3) at (0,-3cm) {};
    \node[dot, fill=white, label={180:$0{.}4$}] (q4) at (3cm,-3cm) {};
    \node[dot, fill=white, label={0:$0{.}9$}] (q9) at (6cm,-3cm) {};
    \node (q) at (9cm,-3cm-3mm) {$\vectors'_{L}$};
    
    \node[dot, label={190:$b_{1}$}] (b1) at (1.5cm,-5cm) {};
    \node[dot, label={-10:$b_{2}$}] (b2) at (4.5cm,-5cm) {};
    \node (B) at (9cm,-5.4cm) {$\mathcal{B}$};
    
    \foreach \q in {3,4,9}{
      \foreach \i in {1,2}{
        \draw[semithick] (q\q) to (b\i);
      }
    }

    \node[mybox] at ($(p4.center)+(0,3mm)$) {};    
    \node[mybox] at ($(q4.center)+(0,-2mm)$) {};    
    \node[mybox] at (3cm,-5.4cm) {};    

    \draw[semithick] (p3) -- (q3) node[midway,left]  {\footnotesize$\{\firstColor{1},\secondColor{2}\}$};
    \draw[semithick] (p4) -- (q4) node[midway,left]  {\footnotesize$\{\firstColor{1},\secondColor{2}\}$};
    \draw[semithick] (p9) -- (q9) node[midway,left]  {\footnotesize$\{\firstColor{1}\}$};

    \draw[semithick] (p3) -- (p4) node[midway,above]  {\footnotesize$\{\firstColor{1},\secondColor{2}\}$};
  \end{tikzpicture}
    \caption{Construction of the graph $G$ for a \acl{BP} instance with sets $\vectors_{S}=\{0{.}1, 0{.}15, 0{.}2\}$ and $\vectors_{L}=\{0{.}3, 0{.}4, 0{.}9\}$.
    The guessed number of bins is $C=3$.
    All small items are packed separately and the bin containing $0{.}15$ is finalized ($C_{0}=1$).
    There is thus a bin containing $0{.}1$ associated with \firstColor{color $1$} (the first value in the braces) and a bin
    containing $0{.}2$ associated with \secondColor{color $2$} (the second value in the braces).
    The color $\bot$ used between all nodes of $\vectors'_{L}$  and all nodes of $\mathcal{B}$ are omitted. 
  }\label{fig:example}
\end{figure}
Finally, we define the color function $\lambda$ with $\lambda \colon E \rightarrow 2^{\colors\cup \{\bot\}}$, such that each edge in $E_c$ gets color $c$ for each $c \in \colors' :=\colors\cup \{\bot\}$.
More formally, we define $\lambda(e) := \{c \in \colors \cup \{\bot\}\,|\, e \in E_{c}\}$.
See \Cref{fig:example} for an example of the construction.
Note that the weights on the edges are irrelevant in this case and can be set to
one, i.\,e.~$\gamma(e,c)=1$.
To finalize the reduction, we have to define the size $\ell$ of the matching we are looking for.
We aim to find a perfect matching and hence are searching for a matching of size $\ell := |\vectors_L| + C_2$. Note that if $C_2 = 0$ and therefore no blocker nodes are introduced, we also remove the color $\bot$ from the set of colors. 

\begin{lemma}
\label{l:ReductionVC}
  There is a packing of the large vectors $\vectors_L$ inside $(C-C_0)$ containers such that each container holds at least one large vector if and only if the above described instance for {\sc Perfect Over-The-Rainbow-Matching} is a ``yes''-instance.
\end{lemma}
\begin{proof}
  Assume there is a packing of the vectors $\vectors_L$ inside $(C-C_0)$ containers such that each container holds at least one large vector.
  In this case, we can construct a perfect over-the-rainbow matching $M$ as follows.
  For each pair of vectors $v,w \in \vectors_L$ that is assigned to the same container, we choose the corresponding edge $\{v,w\}$ for the matching and assign it the corresponding color $c \in \colors$.
  For each vector $v  \in \vectors_L$ that is the only large vector in its container, we choose the edge $\{v,v'\}$ for the matching and assign it the corresponding color $c \in \colors$. 
  To this point all the vectors in $\vectors_L$ are covered by exactly one matching edge since each of them is contained in exactly one container.
	
  Note that in the given packing there have to be exactly $C_1 = 2(C-C_0) - |\vectors_L| $ containers with exactly one large vector and $C_2 = |\vectors_L| - (C-C_0)$ containers with exactly two large vectors.
  As a consequence, there are exactly $2 \cdot C_2$ nodes in $\vectors_L'$ that are not yet covered by a matching edge since their originals are covered by edges between each other.
  For each of these nodes, we choose an individual node from the set $\mathcal{B}$ and define the edge between these nodes as a matching edge and assign it the color $\bot$. 
  Since there are exactly $2 \cdot C_2$ blocker nodes, we cover all nodes in $V$ with matching edges and hence we have constructed a perfect matching.
  Each color $c \in \colors \setminus \{\top\}$ is represented by one partially filled container and hence each has to appear in the matching.
  Moreover, if the color $\top$ was introduced, that is, there were less than $C-C_0$ containers partially covered by small vectors, then there was a container exclusively containing large vectors and hence $\top$ was used in the matching as well.
  Therefore, we indeed constructed a perfect over-the-rainbow matching.

  Conversely, assume that we are given a perfect over-the-rainbow matching $M$.
  Consequently, each vector in $\vectors_L$ is covered by exactly one matching edge. 
  As $M$ contains at most $|\vectors_L| + C_2$ edges, and $2 \cdot C_2$ edges are needed to cover the nodes in $\mathcal{B}$, there are exactly $|\vectors_L| - C_2 = (C-C_{0})$ matching edges containing the nodes from $\vectors_L$.
  As in $M$ each color is present, we can represent each container by such a matching edge and place the corresponding vector or vectors inside corresponding containers.
  If a color $c \in \colors \setminus \{\top\}$ appears more than once, we use an empty container for the corresponding large vectors.
\end{proof}

To decide if there is a packing into at most $C$ containers, we find a partition of the $k$ small vectors with $\cO(k!)$ guesses, and the to-be-discarded containers with $\cO(2^k)$ guesses.
Constructing the graph~$G$ needs $\cO(n^2k)$ operations.
By \Cref{thm:over-the-rainbow-algo}, a perfect over-the-rainbow matching over $k+O(1)$ colors with weight $\ell \in O(n)$ can be computed in time $2^k\cdot n^{O(1)}$.
To find the correct $C$ we call the above algorithm in binary search fashion $\cO(\log(n))$ times, as we need at most~$n$ containers.
This results in a run time of $2^{2k}\cdot k!\cdot n^{O(1)}= 4^k\cdot k!\cdot n^{O(1)}$.

\iflongversion
\subparagraph*{The Case of Covering Vectors}
Recall that in the \acl{pVC} problem, we are given $n$ vectors $\vectors$ of dimension $d$ and a set of containers, each with the same size limitation $T \in \QQ^d$. 
Further, we are given a partition of the vectors $\vectors$ into the set $\vectors_L$ of large vectors and and the set $\vectors_S$ of small ones.
The large vectors have the property that every subset of three vectors cover a container.
On the other hand, we need at least two vectors to cover one container (otherwise, we can remove the corresponding vectors from the instance and only consider the residual instance).
Hence an optimal solution covers at least $\lfloor \vectors_L/3 \rfloor$ containers, while it can cover at most $\lfloor \vectors/2 \rfloor$ containers.
Remark that for this problem we can search for the optimal number of covered containers in binary search fashion between the bounds $\lfloor \vectors_L/3 \rfloor$ and $\lfloor \vectors/2 \rfloor$: 
A solution covering a given number of containers can be transformed into a solution covering one less container.
In the following, we assume we are given the number $C$ of containers to be covered and have to decide whether this is possible or not.
Furthermore, note that each solution which contains multiple partially covered containers can be transformed into a solution where each container is completely covered, by distributing the vectors from the non-covered containers to the covered containers. 
Clearly, this might empty some containers completely. 

Similar as for \acl{pVP}, we first guess the distribution of small vectors $\vectors_S$ to the (at most $C$) containers.
Each distribution of these vectors affects at most $k$ containers, as $|\vectors_S| = k$.
Since the order of the containers is irrelevant, there are at most $\cO(k!)$ distinct possibilities to distribute these small vectors.

In the next step, we remove all containers that are completely filled by now.
This leaves $C'\leq C$ containers we have to cover.
Let $C_{p}\leq C'$ be the number of those containers that contain a small vector from $\vectors_{S}$.
We call these containers \emph{partially covered}.
As in the algorithm for the \acl{pVP} problem, we introduce a set of colors $\colors$ such that each partially covered container is represented by one color $c \in \colors$ and the empty container is represented by one color called $\top$.
On the one hand, as we have to distribute the vectors in $\vectors_L$ to $C'$ containers, there are at least $|\vectors_L| - 2\cdot C'$ containers with more than two large vectors.
On the other hand, there are $C_P$ partially covered containers, and they might need only one vector to be covered while all others need at least two vectors to be covered. 
Hence, the number of containers admitting more than two vectors is bounded by $|\vectors_L| - 2\cdot (C'-C_P) -C_P = |\vectors_L| - 2\cdot C'+C_P$.
Note that a container with more than three large vectors stays covered if one of the large vectors is removed. 
Hence, we can guarantee that if containers with at most two vectors exists, there are no containers with more than three large vectors. 

In the next step, we guess the number $C_1$ of containers with only one large vector in the optimal solution.
As a consequence, there are exactly $|\vectors_L| -C_1 $ large vectors that have to be placed inside containers with more than one large vector and $C' - C_1$ such containers.
Consequently, there are exactly $C_3 := (|\vectors_L| -C_1 ) - 2(C'-C_1)$ containers with three large vectors and hence $C_2 := C' - C_1 -C_3$ container with exactly two large vectors.
Note that for this guess there are at most $\cO(k)$ options since only partially filled containers can be covered by only one large vector.

Knowing these values, we construct a graph $G=(V,E)$ for a \acl{pPOTRM} problem as follows.
Similar as above the set of nodes is a combination of nodes generated for the large vectors and some blocker nodes.
Again each vector $v \in \vectors_{L}$ has a node in the graph as well as a copy of itself $v' \in \vectors_{L}'$. 
Furthermore, we introduce two sets of blocker nodes $\mathcal{B}_1$ and $\mathcal{B}_2$ such that $|\mathcal{B}_1| = |\vectors_{L}| -  C_1$ and $|\mathcal{B}_2| = |\vectors_{L}| -  (C_1 + 2\cdot C_2)$.

For each pair of nodes $v,w \in \vectors_{L}$, we introduce one edge $e = \{v,w\}$ and assign the color~$c$ to it if these two vectors together cover the corresponding container that includes the small vector associated with color $c$. 
Similarly, we introduce an edge $e = \{v,v'\}$ between a vector $v \in \vectors_{L}$ and its copy $v' \in \vectors_{L}'$ and assign it with the color $c \in \colors$ if this vector alone covers the corresponding container.
More precisely, we introduce for each $c \in \colors$ the edge set  $E_c := \{\{u,v\}\, | \, u,v \in \vectors_L, s(u) + s(v) \geq s(c)\} \cup \{\{v,v'\}\,|\, v \in \vectors_L, v+v' \geq s(c)\}$.
Additionally, we introduce all edges between the copied nodes $\vectors_L'$ and the blocker nodes~$\mathcal{B}_1$ to ensure that exactly $C_1$ vectors are placed alone inside their containers.
Furthermore, we introduce all edges between the blocker nodes $\mathcal{B}_2$ and the vector nodes $\vectors_L$, to ensure that exactly $C_1 + 2 \cdot C_2$ vectors are placed inside the containers.
More formally, we define $E_{\bot} := \{\{v',b\} \,|\, v' \in \vectors_L', b \in \mathcal{B}_1\} \cup \{\{v,b\} \,|\, v \in \vectors_L, b \in \mathcal{B}_2\}$ where $\bot \not \in \colors$. 
Together, we get $E := E_{\bot} \cup \bigcup_{c \in \colors} E_c$.

Finally, we have to define the color function $\lambda$, the weight function $\gamma$ and the maximal weight of the matching $\ell$.
We define $\lambda \colon E \to 2^{\colors\cup \{\bot\}}, e \mapsto \{c \,|\, c \in \colors \cup \{\bot\}, e \in E_c  \}$ and
\begin{equation*} 
  \gamma(e,c) = \begin{cases} 1, & \text{if } c \in\colors \setminus \{\bot,\top\},\\ 0, &\text{otherwise,} \end{cases}
\end{equation*}
for all $e = \{v,w\} \in E$ and $c \in \colors$ with $c \in \lambda(e)$.
We want to allow each color for a partially filled container to be taken at most once and, hence, search for a matching with weight at most $\ell := |\colors \setminus \{\bot,\top\}|$.


\begin{lemma}
  There is a covering for the $C'$ containers using the large vectors $\vectors_L$ such that there are exactly $C_1$ containers with one large vector and $C_2$ containers with two large vectors if and only if the above described graph $G$ has a perfect over-the-rainbow matching $M$ of weight $\ell = |\colors \setminus \{\bot,\top\}|$.
\end{lemma}
\begin{proof}	
  Assume there is a packing of the vectors $\vectors_L$ inside the $C'$ container, such that there are exactly $C_1$ containers with one large vector and $C_2$ containers with two large vectors.
  In this case, we can construct a perfect over-the-rainbow matching $M$ as follows.
  For each pair of vectors $v,w \in \vectors_L$ that is assigned to the same container, we choose the corresponding edge $\{v,w\}$ for the matching and assign it the corresponding color.
  For each vector $v  \in \vectors_L$ that is the only large vector in its container, we choose the edge $\{v,v'\}$ for the matching and assign it the corresponding color. 
  To this point exactly $C_1$ nodes from the set $\vectors_L'$ are covered by the matching and exactly $C_1 + 2\cdot C_2$ nodes from the set $\vectors_L$ are covered by matching edges.	

  As a consequence there are exactly $|\vectors_L| - C_1$ nodes in $\vectors_L'$ that still need to be covered and exactly $|\vectors_L| - (C_1 + 2 \cdot C_2)$ nodes in $\vectors_L$ that need to be covered.
  For each of the $|\vectors_L| - C_1$ nodes in $\vectors_L'$, we choose one individual node from the $|\vectors_L| - C_1$ nodes in $\mathcal{B}_1$ arbitrarily, add the corresponding edge to the matching $M$ and assign it the color $\bot$.
  For the $|\vectors_L| - (C_1 + 2 \cdot C_2)$ nodes in $\vectors_L$, we choose one individual node from the $|\vectors_L| - (C_1 + 2 \cdot C_2)$ nodes in $\mathcal{B}_2$ arbitrarily and add the corresponding edge to the matching $M$ and assign it the color $\bot$.
	
  Obviously $M$ covers all the nodes in the graph and is a matching and hence $M$ is a perfect matching.
  Furthermore, each color $\colors \setminus \{\bot,\top\}$ is chosen exactly once while the color $\top$ is chosen at least once.
  Since the edges with colors in $\colors \setminus \{\bot,\top\}$ have a weight of $1$ while all other edges have a weight of $0$, $M$ has a weight of exactly  $|\colors \setminus \{\bot,\top\}| = \ell$.
	
  On the other hand, assume that we are given a perfect over-the-rainbow matching $M$ for $G$ with weight exactly $\ell = |\colors \setminus \{\bot,\top\}|$.
  Since all the colors have to appear at least once, and each edge with color in $\colors \setminus \{\bot,\top\}$ has a weight of exactly $1$ each of these colors can appear at most once.
  Since $M$ is a perfect matching, each vector in $\vectors_L$ is covered by exactly one matching edge and $M$ contains exactly $|V|/2 = 2\dot |\vectors_L| - (C_1 + C_2)$ edges.
  Exactly $|\mathcal{B}_1| = |\vectors_L| - C_1$ nodes in $\vectors_L'$ are contained in edges between $\vectors_L'$ and $\mathcal{B}_1$, since all the neighbors of nodes in $\mathcal{B}_1$ can be found in $\vectors_L'$.
  Hence, there are exactly $|\vectors_L'| - |\mathcal{B}_1| = C_1$ edges with colors $c \in \colors$ in $M$ that are between the nodes in $\vectors_L$ and their corresponding copies in~$\vectors_L'$.
  We place the corresponding vectors alone in the container corresponding to the color of the edge.  
  The remaining $|\vectors_L| - C_1$ nodes in $\vectors_L$ are covered by edges between each other or by edges to the set $\mathcal{B}_2$. 
  Since $|\mathcal{B}_2| = |\vectors_L| - (C_1 + 2 \cdot C_2)$ and these nodes only have neighbors in the set $\vectors_L$ there are exactly $|\vectors_L| - (C_1 + 2 \cdot C_2)$ nodes in $\vectors_L$ that share a matching edge with a node in $\mathcal{B}_2$. 
  Hence the remaining $(|\vectors_L| - C_1) - (|\vectors_L| - (C_1 + 2 \cdot C_2)) = 2 \cdot C_2$ nodes have to be paired by matching edges.
  We place these vectors pairwise inside the corresponding containers.
  The residual vectors are distributed in groups of three and place inside the residual $C' - C_1- C_2$ container.
  Since each of the colors for the partially filled containers appears exactly once in the matching, all the containers are covered by this assignment.
\end{proof} 

In the following, we summarize the run time of the above described algorithm to decide whether there is a packing into at most
$C$ containers.
Finding the correct partition of the $k$ small vectors can be done in $\cO(k!)$ guesses.
Finding the number of containers with at most one large vector can be done in $\cO(k)$ guesses.
Finally, the construction of the graph $G$ needs at most $\cO(n^2k)$ operations.
Hence the run time can be summarized as $2^k\cdot k!\cdot n^{O(1)}\leq 4^k\cdot k!\cdot n^{O(1)}$.
Lastly, to find the correct $C$ we have to call the above algorithm in binary search fashion at most $\cO(\log(n))$ times, since we can cover at most $n$ containers.


\subparagraph*{The Case of Packing Vectors with Profits}
Recall that in the \acl{pVMKP} problem, we are given a set $\vectors$ of $n$ vectors with dimension $d$, a profit function $p\colon \vectors \to \NN_{\geq 0}$ and $C$ containers each with capacity constraint $T \in \QQ^d$.
Furthermore, we are given a partition of the vectors $\vectors$ into small $\vectors_{S}$ and large $\vectors_{L}$.

Again, we guess the distribution of the small vectors. 
However, since it might not be optimal to place all the small vectors, we first have to guess which subset of them is chosen in the optimal solution. 
There are at most $2^{k}\cdot k!$ possibilities for both guesses.
After this step, we have at most $k$ containers which are partially filled with small vectors.

In the next step, we guess for each partially filled containers whether they contain an additional large vector and discard the containers that do not.
There are at most $2^k$ possible choices for this.
Let $C_0$ be the number of such discarded containers. 
This step leaves $C-C_0$ containers for the large vectors in $\vectors_{L}$.
Again, we define a color for each remaining partially filled container and one color $\top$ for the empty containers resulting in a set $\colors$ of at most $k+1$ colors.

Similar as for the problems \acl{pVP} and \acl{pVC}, we construct a graph $G = (V,E)$ to find the profit maximal packing.
We introduce one node for each vector in $v \in \vectors_{L}$ and a node for its copy $v' \in \vectors_{L}'$.
Furthermore, we introduce a set $\mathcal{B}$ of $2\cdot |\vectors_{L}| - 2 \cdot (C-C_0)$ blocker nodes to ensure that we use exactly $C-C_0$ container.
We define a profit of zero for the copy nodes and the blocker nodes, while the nodes for the original vectors $v \in \vectors_{L}$ have profit $p(v)$.

We add an edge between two nodes $v,w \in \vectors_{L}$ and assign it the color $c \in \colors$ if the vectors together fit inside the corresponding container assigned with color $c$.
Furthermore, we add an edge between a node $v \in \vectors_{L}$ and its copy $v' \in \vectors_{L}'$ and assign it the color $c  \in \colors$ if it fits alone inside the corresponding container.
More formally, we define for each color $c \in \colors$ the set $E_c := \{\{u, v\}\, | \, u,v \in \vectors_L, s(u) + s(v) \leq s(c)\} \cup \{\{v,v'\}\,|\, v \in \vectors_L, v+v' \leq s(c)\}$.
Finally, we connect each node from the set $\mathcal{B}$ with each node from the set $\vectors_L \cup \vectors_L'$, i.e, we define $E_{\bot} := \{\{v,b\}\,|\, v \in \vectors_L \cup \vectors_L', b \in \mathcal{B}\}$. 
In total, we set $E := E_{\bot} \cup \bigcup_{c \in \colors} E_c$.

Finally, we define the color function $\lambda$ and the profit function $\gamma$.
For this purpose we denote $p_{\max} := \max\{p(v) | v \in \vectors_L\} $ and define $\lambda \colon E \rightarrow 2^{\colors\cup \{\bot\}}, e \mapsto \{c\,|\, c \in \colors \cup \{\bot\}, e \in E_c \}$ as well as
\begin{equation*}
  \gamma(\{v,w\},c) := \begin{cases}
                         2p_{\max} - (p(v) +p(w)), & v,w \in \vectors_L \cup \vectors_L',\\
                         0, & \text{otherwise,}
                       \end{cases}
\end{equation*}
for all $e = \{v,w\} \in E$ and $c \in \colors$ with $c \in \lambda(e)$.
Obviously, all the weights are non-negative.

%

\iflongversion
\begin{lemma}
\else
\begin{lemma}[$\star$]
\fi
  There is a packing of the large vectors inside the corresponding $C-C_0$ containers with profit at least $p$ if and only if there is a perfect over-the-rainbow matching $M$ in $G$ with weight at most $(C-C_0)p_{\max} - p$.
\end{lemma}
\iflongversion
\begin{proof}
  Assume, we are given a packing of large vectors inside the $C-C_0$ containers with profit at least $p$.
  For each packing of large vectors inside one container, we choose the edge between the corresponding pair of vectors (or between the vector and its copy in the case that the container has only one vector) for the matching and assign it the corresponding color. 
  Now there are exactly $2\cdot(C-C_0)$ nodes in $\vectors_L \cup \vectors_L'$ covered by the matching.
  The remaining $|\vectors_L \cup \vectors_L'| - 2\cdot(C-C_0)$ nodes in $\vectors_L \cup \vectors_L'$ are paired with one arbitrary node in $\mathcal{B}$. Since $\mathcal{B}$ contains exactly $2|\vectors_L| - 2\cdot(C-C_0)$ nodes, each node can be paired.
	
  The obtained matching $M$ is a perfect matching, since each node is covered. 
  Furthermore, each color in $\colors \setminus \{\top \}$ is used exactly once by definition of the colors.
  Let $\vectors_{L,S} \subseteq \vectors_L$ be the set of large vectors packed in the given solution.
  By definition of the solution it holds that $p(\vectors_{L,S}) \geq p$.
  Note that the weight of an edge between two nodes $v,w \in \vectors_L$ is given by $2p_{\max} - (p(v) +p(w))$, while edges between a node $v \in \vectors_L$ and its copy $v'$ have the weight $2p_{\max} - p(v)$. 
  All the edges to the blocker nodes $\mathcal{B}$ have weight $0$.
  Hence the weight of the matching is given by $2(C-C_0)p_{\max} - p(\vectors_{S}) \leq 2(C-C_0)p_{\max} - p$, which proves the first implication.
	
  To prove the other direction, assume that we are given a perfect over-the-rainbow matching with weight at most $2(C-C_0)p_{\max} - p$.
  Each of the $2|\vectors_L| - 2\cdot(C-C_0)$ blocker nodes in~$\mathcal{B}$ is matched to exactly one node in $|\vectors_L \cup \vectors_L'|$. 
  As a result there are exactly $2\cdot(C-C_0)$ nodes in $|\vectors_L \cup \vectors_L'|$ that are paired by the matching.
  We place the corresponding vectors inside the corresponding containers with regard to the color of the matching edge. 
  If a color $c \in \colors \setminus \{\top\}$ appears more than once, we use an empty container.
	
  This packing is valid, since each color appears at least once and hence we can fill each container.
  Let $\vectors_{L,S} \subseteq \vectors_L$ be the set of large vectors that are matched with a node from the set $\vectors_L \cup \vectors_L'$.
  Then, by definition of the weight function, the matching has a size of $2(C-C_0)p_{\max} - \sum_{v \in \vectors_{L,S}}p(v) \leq 2(C-C_0)p_{\max} - p$.
  As a consequence the profit of the packing is given by  $\sum_{v \in \vectors_{L,S}}p(v) \geq p$.
\end{proof}
\fi

We can summarize the steps of the algorithm as follows. 
For each choice of small items and each possibility to distribute these items, the algorithm considers each choice of partially filled containers that do not contain an additional large item.
For each of these choices the algorithm constructs the graph described above.
Then it performs a binary search for a perfect over-the-rainbow matching with the smallest possible weight in the bounds $[ 0,2(C-C_0)p_{\max}]$.
Finally, it returns the packing with the largest total profit found among all possibilities. 

By \Cref{thm:over-the-rainbow-algo}, we need at most $2^{|\colors|} \cdot n^{\cO(1)} \cdot (p_{\max})^{O(1)} = 2^{k+2} \cdot n^{\cO(1)}\cdot (p_{\max})^{O(1)}$ operations to solve the constructed \acl{POTRM} problem.
Finding the correct choice and partition of the $k$ small vectors can be done in $\cO(k \cdot k!)$ guesses. 
To find the containers without a large vector can be done in $\cO(2^k)$ guesses.
Finally the construction of the graph $G$ needs at most $\cO(n^2k)$ operations.
The binary search procedure over the profits can be done in at most $\cO(\log((C-C_0)p_{\max})) = \cO(\log(n)+\log(p_{\max}))$ operations, since the number of containers is bounded by $n^{\cO(1)}$.
Hence, the run time is $2^{2k+2}\cdot k! \cdot n^{\cO(1)}\cdot (p_{\max})^{O(1)} =  4^{k}\cdot k! \cdot n^{\cO(1)}\cdot (p_{\max})^{O(1)}$.
\fi

\section{Find Over-the-Rainbow Matchings with Conjoining Matchings}
\label{sec:solve_otrpm}
\iflongversion
In the previous section, we have reduced several packing and covering problems to \acl{pPOTRM} problems.
\fi
Of course, all this effort would be in vain without the means to find such matchings.
This section presents a reduction to the task of finding a conjoining matching, which results in a parameterized algorithm for finding perfect over-the-rainbow matchings by applying \Cref{t:GutinCGM}.
Overall, this proves \Cref{thm:over-the-rainbow-algo}, which is repeated below for convenience:

\par\medskip\noindent\textbf{\sffamily Claim of \Cref{thm:over-the-rainbow-algo}.}
\bgroup\itshape\ignorespaces
\claimedtheorem
\egroup\smallskip

We aim to construct graphs $H'$ and $G'$ such that $G$ has a perfect over-the-rainbow matching if, and only if, $G'$ has a perfect conjoining matching with respect to $H'$ of the same weight.
Recall that in an over-the-rainbow matching, we request an edge of every color to be part of the matching; while in a conjoining matching, we request edges between certain sets of nodes to be part of the matching.
For the reduction, we transform $G$ into $G_1,\dots,G_{\left|\mathcal{C}\right|}$ where each~$G_{c}$ is a copy of $G$ containing only edges of color $c$. 
Hence, $V(G_{c})=\{v_c\mid v \in V(G)\}$, i.e., $v_c$ is the copy of $v \in V(G)$ in $V(G_c)$, and $G_{c}$ contains only edges $e\in E(G)$ with $c\in\lambda(e)$. 
We set $G'$ to be the disjoint union of the~$G_c$ while setting $V(H')=\{\, V(G_c) \mid c\in\mathcal{C}\,\}$ and $E(H')=\{\,\{h,h\} \mid h\in V(H')\,\}$. 
Now a conjoined matching contains an edge of every color---however, the same edge of~$G$ could be used in multiple ways in the different copies $G_c$.

To address this issue, we introduce a gadget that will enforce any perfect matching in~$G'$ to use at most one copy of every edge
of~$G$.
In detail, for every node $v\in V(G)$ we will add an independent set $J(v)$ of size $\left|\mathcal{C}\right|-1$ to~$G'$.
Furthermore, we will fully connect $J(v)$ to all copies of $v$ in $G'$, that is, we add the edges $\{v_c,x\}$ for all $c\in\mathcal{C}$ and $x\in J(v)$ to $G'$.
This construction is illustrated in \Cref{fig:edge-blocker-gadget}.
Observe that in any perfect matching of $G'$ all elements of $J(v)$ must be matched and, thus, we ``knock-out'' $\left|J(v)\right|=\left|\mathcal{C}\right|-1$ copies of~$v$ in~$G'$---leaving exactly one copy to be matched in one $G_c$.
We add one more node to $H'$ that represents the union of all the sets $J(v)$ and has no connecting edge.
To complete the description of the reduction, let us describe the weight function of $G'$: For each $e\in E(G')$, we define
\iflongversion
\begin{equation*}
	\gamma'(e) : =\begin{cases}
		\gamma(e,c) & \text{if $e \in E(G_c)$ for some $c\in\mathcal{C}$} \\
		0 & \text{otherwise.} 
	\end{cases}
\end{equation*}
\fi
Note that this definition implies that $\gamma'(e)=0$ for each $e$ with $e\cap J(v)\neq\emptyset$ for some $v \in V(G)$.
\begin{figure}
  \begin{center}
    \begin{tikzpicture}[dot/.style={circle,draw=black,fill=black,inner sep=1pt},
      node distance = 0.10cm]
      \begin{scope}
        \node[anchor=north west] (Glabel) at (-0.75,1.25) {\small$G:$};
        \node[dot, label={[label distance=0.33mm]190:$v$}] (v) at (0,-0.03) {};
        \draw[gray,green!75!black]  (v) -- ++(-0.4,.6);
        \draw[gray,blue!75!black]   (v) -- ++(-0.2,.6);
        \draw[gray,red!75!black]    (v) -- ++(-0.0,.6);
        \draw[very thick,red!75!black]   (v) -- ++(+0.2,.6);
        \draw[gray,blue!75!black]   (v) -- ++(+0.4,.6);
        \draw[gray,orange!75!black] (v) -- ++(+0.6,.6);
      \end{scope}	
      \begin{scope}[xshift=1cm,yshift=-.03]
        \node[anchor=north west] (Gdlabel) at (1.0,1.25) {\small$G':$};
        \foreach \i/\ax/\bx/\cx/\col in {1/-0.2/0.4/0.4/blue,2/-0.4/-0.4/-0.4/green,3/0/0/+0.2/red,4/0.6/0.6/+0.6/orange}{
          \node[dot, label={[label distance=0.33mm]177.5:$v_{\i}$}] (x\i) at (\i*1.5,0){};    
          \draw[gray,\col !75!black] (x\i) -- ++(\ax,.6);
          \draw[gray,\col !75!black] (x\i) -- ++(\bx,.6);
          \draw[gray,\col !75!black] (x\i) -- ++(\cx,.6);
        }
        \foreach \i in {1,2,3}{
          \node[dot] (y\i) at (\i*1.5+0.75,-1){};
        }
        \foreach \i in {1,2,3,4}{
          \foreach \j in {1,2,3}{
            \draw[] (x\i) -- (y\j);
          }
        }
        \draw[very thick] (x1) -- (y1);
        \draw[very thick] (x2) -- (y2);
        \draw[very thick] (x4) -- (y3);
        \draw[very thick,red!75!black] (x3) -- ++(+0.2,.6);
        \node[draw, semithick,
              rounded corners, densely dashed,
              inner sep=0pt, minimum width=3.5cm,
              minimum height=0.5 cm, color=darkgray] at (y2.center) {};
         \node[baseline] at ($(y1)+(-0.75,.5ex)$) {\small$J(v)$};
      \end{scope}
    \end{tikzpicture}
  \end{center}
  \caption{
    Reducing the problem of finding a perfect over-the-rainbow matching to the problem of finding a perfect conjoining matching.
    Left: single node of the colored input graph with a thick edge from a perfect matching.
    Right: $|\mathcal{C}|=4$ copies of $v$ in $G'$; the corresponding subgraphs $G_c$ only contain edges of a single color.
    At the bottom, the added set $J(v)$ which is fully connected to all copies of $v$.
    The thick edges indicate how these nodes are paired in a perfect matching.
  }
  \label{fig:edge-blocker-gadget}
\end{figure}
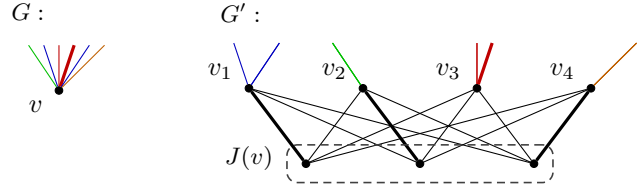
\vspace{-1em}
\iflongversion
\begin{lemma}
\else
\begin{lemma}[$\star$]
\fi
\label{lem:rainbow-conjoining-matching}  
  Let $\big(G,\lambda,\gamma\big)$ be a colored and edge-weighted graph, and let $G'$ and $H'$ be defined as above.
  There is a perfect over-the-rainbow-matching $M$ of weight $\ell$ in $G$ if, and only if, there is a perfect conjoining matching $M'$ of weight $\ell$ in $G'$.
\end{lemma}
\iflongversion
\begin{proof}
   First, let us consider a perfect over-the-rainbow matching $M$ in $G$.
   Let $\xi\colon M\rightarrow\mathcal{C}$ be a surjective function with $\xi(e)\in\lambda(e)$ for all $e\in E(G)$.
   We have to show that there is a perfect conjoining matching in $G'$.
   Let $M'=\{\{v_{\xi(e)},w_{\xi(e)}\}\mid e=\{v,w\}\in M\}$.
   Since~$M$ is a matching in $G$, $M'$ is a matching in $G'$; and since $\xi$ is surjective we have that $M'$ contains at least one edge in every copy $G_c$ of $G$ in $G'$ and, thus, $M'$ is actually a conjoining matching.
   Furthermore, there is a bijection~$f(e)=e_{\xi(e)}$ between $M$ and $M'$ such that $\gamma(e,\xi(e))=\gamma'(f(e))$, which implies that the total weight of both matchings is the same. 

   By the definition of $M'$, for every node~$v\in V(G)$, there is exactly one color~$c\in\mathcal{C}$ for that there is an edge
   in~$M'$ containing $v_c$.
   Therefore, the set $\{v_1,\dots,v_{|\mathcal{C}|}\}$ contains exactly $|\mathcal{C}|-1$ unmatched nodes for all $v\in V(G)$.
   We conclude that $M'$ can be extended to a perfect conjoining matching $M''$ by paring these nodes with $J(v)$.
   Observe that $M''$ has the same weight as $M'$ as the added edges have weight zero.
   Therefore, $M''$ has the same weight as $M$.

   For the other direction, let us consider a perfect conjoining matching $M'$ in~$G'$.
   Observe that for all nodes $v\in V(G)$ the nodes in $J(v)$ have to be matched by $M'$ and, thus, for all nodes $v\in V(G)$ there is exactly one node $\alpha(v)\in\{v_1,\dots,v_{\left|\mathcal{C}\right|}\}$ that is \emph{not} matched with an element of $J(v)$.
   We define the set $M=\{\,\{v,w\}\mid\text{$v,w\in V(G)$ and $\{\alpha(v),\alpha(w)\}\in M'$}\,\}$ and claim that $M$ is a perfect over-the-rainbow matching of the same weight as $M'$.
   First observe that all $v\in V(G)$ are matched by $M$, since $M'$ is a perfect matching and, thus, matches $\alpha(v)$ with, say, $w_i$.
   Observe that by the definition of $\alpha$ we have $w_i\not\in J(v)$ and by the construction of $G'$ we have that $w_i$ is a copy of some $w\in V(G)$ (it can, in particular, not be part of any other $J(u)$).
   Since $w_i$ is paired with $\alpha(v)$, we conclude $\alpha(w)=w_i$ and, thus, $\{\alpha(v),\alpha(w)\}\in M'$.
   Further, notice that every $v\in V(G)$ can be matched by at most one element of $M$, as $M'$ is a perfect matching and, thus, matches $\alpha(v)$ with exactly one other node.
   We conclude that $M$ is a perfect matching of $G$.
   Finally, for all $\{v,w\}\in M$ observe that $\{\alpha(v),\alpha(w)\}$ must lie in some copy $G_c$ of $G$ in $G'$.
   We define $\xi(\{v,w\})=c$ and $\gamma(\{v,w\},c)=\gamma'(\{\alpha(v),\alpha(w)\})$.
   Observe that $\xi$ is surjective since $M'$ is conjoining and, thus, witnesses that $M$ is a perfect over-the-rainbow matching of $G$.

   To conclude the proof, notice that $M$ has the same weight as $M'$ as for any edge of $M'$ that has non-zero weight (that is, any edge that is not connected to some $J(v)$), we have added exactly one edge of the same weight to $M$.   
 \end{proof}
\fi
\begin{proof}[Proof of \Cref{thm:over-the-rainbow-algo}]
  Let $\big(G, \lambda, \gamma, \ell\big )$ be an instance of \acl{pPOTRM}.
  We construct in polynomial time an instance $(G',H',\gamma',\ell)$ of \acl{pCM}, where the partition of $V(G')$ is defined as $V(G_1)\dotcup V(G_2)\dotcup\dots\dotcup V(G_{\left|\mathcal{C}\right|})\dotcup J$ with $J=\bigcup_{v\in V(G)}J(v)$ and $E(H')$ contains one self loop for each $V(G_c)$, $c \in \colors$.  
  By \Cref{lem:rainbow-conjoining-matching}, $(G',H',\gamma',\ell)$ has a perfect conjoining matching of weight $\ell$ if, and
  only if $\big(G, \lambda, \gamma, \ell\big )$ has a perfect over-the-rainbow matching of weight $\ell$.
  We apply \Cref{t:GutinCGM} to find such a conjoining matching in time \smash{$2^{|E(H')|}\cdot n^{O(1)}\cdot \ell^{O(1)}$}.
  Observe that $|E(H')|=|\colors|$ and, thus, we can find the sought perfect over-the-rainbow matching in time
  \smash{$2^{\left|\colors\right|}\cdot n^{O(1)}\cdot \ell^{O(1)}$}.
\end{proof}

\section{A Deterministic Algorithm for Bin Packing with Few Small Items}
\label{sec:deterministic}
We now present a \emph{fully-deterministic algorithm} for \acl{pBP}.
The price we have to pay for circumventing the randomness is an increased run time as we avoid the polynomial identity testing subroutine.
On the bright side, this makes the algorithm straightforward and a lot simpler.
We anticipate that extending this algorithm for \acl{pVP} seems quite challenging.
The main obstacle here is to identify the maximum item size in some sets, a task for which there does not seem to be a  sensible equivalent notion for vectors.

\vspace{-1em}
\subparagraph*{About the Structure of Optimal Solutions.}
In the following, we prove the existence of an optimal solution that admits some useful properties regarding the placement of large items relating to small ones. These properties are utilized in the algorithm later on.
\begin{claim}
\label{clm:SmallItemProperty1}
  There exists an optimal solution where the total size of small items on each bin containing only small items is larger than the total size of small items on each bin containing additionally large items.
\end{claim}
\iflongversion    
\begin{figure}
  \centering
  \begin{tikzpicture}
    \pgfmathsetmacro{\h}{2}
    \pgfmathsetmacro{\w}{1}

    \draw (0,0) rectangle (\w,\h);
    \node[below] at (0.5*\w,0) {$b_1$};
    \draw[fill = lightgray] (0,0) rectangle node {$s_1$} (\w,0.2*\h);

    \begin{scope}[xshift = 2*\w cm]
      \draw (0,0) rectangle (\w,\h);
      \node[below] at (0.5*\w,0) {$b_2$};
      \draw[fill = lightgray] (0,0) rectangle node {$s_2$} (\w,0.3*\h);
      \draw[fill = gray] (0,0.3*\h) rectangle node {$l_1$} (\w,0.65*\h);
      \draw[fill = gray] (0,0.65*\h) rectangle node {$l_2$} (\w,\h);
    \end{scope}

    \begin{scope}[xshift = 4*\w cm]
      \draw [->,decorate, decoration={snake,pre length=0pt,post length=1pt}, thick] (0,0.5 *\h) -- (2*\w,0.5*\h);
    \end{scope}

    \begin{scope}[xshift = 7*\w cm]
      \draw (0,0) rectangle (\w,\h);
      \node[below] at (0.5*\w,0) {$b_1$};
      \draw[fill = lightgray] (0,0) rectangle node {$s_2$} (\w,0.3*\h);
    \end{scope}

    \begin{scope}[xshift = 9*\w cm]
      \draw (0,0) rectangle (\w,\h);
      \node[below] at (0.5*\w,0) {$b_2$};
      \draw[fill = lightgray] (0,0) rectangle node {$s_1$} (\w,0.2*\h);
      \draw[fill = gray] (0,0.3*\h) rectangle node {$l_1$} (\w,0.65*\h);
      \draw[fill = gray] (0,0.65*\h) rectangle node {$l_2$} (\w,\h);
    \end{scope}
  \end{tikzpicture} 
  \caption{\label{figure:claim10}Proof of \Cref{clm:SmallItemProperty1}. The light gray rectangles denoted by $s_1$ and $s_2$ represent the load of small items on each of the two bins $b_1$ and $b_2$. The dark gray areas denoted with $L_1$ and $l_2$ represent large items in the bin $b_2$} 
\end{figure}
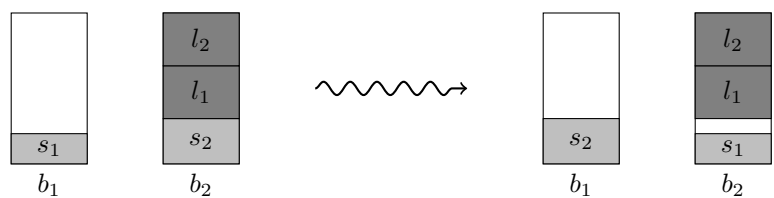
\begin{claimproof}
  Suppose an optimal solution, where the stated property is violated.
  Thus, there exists two bins $b_1$ and $b_2$, where the total size $s_1$ of small items on $b_1$ only admitting small items is smaller than the total size $s_2$ of small items on $b_2$ where also large items are placed ($s_1 \leq s_2$).
  We can now swap the sets of small items in $b_1$ and $b_2$.
  Since $s_1 \leq s_2$, the load of $b_2$ becomes smaller when now containing small items with load $s_1$. 
  On the other hand, the total load on $b_1$ is now $s_2$.
  Since this entire set was placed on one bin before, $b_1$ is not over packed.
  We can iterative repeat this step until the property is satisfied for all bins.
  The proof is illustrated in \Cref{figure:claim10}.
\end{claimproof}
\fi

\begin{claim}
\label{clm:SmallItemProperty2}
  Given an optimal solution and an arbitrary order of the bins containing small items and exactly one large item. We can repack these large items correctly using a largest fitting approach with respect to the order of the bins.
  In detail, we place greedily the largest fitting item into the current bin.
\end{claim}
\iflongversion 
\begin{claimproof}
  Consider the bins containing small items and exactly one large item in the given order. 
  If the current bin $b$ contains the largest fitting item regarding all items being packed on the later bins regarding the order, we consider the next bin.
  Otherwise, we swap the item $i_b$ inside this bin with the largest item $i_{\max}$ that fits inside this bin and was not placed inside a bin which was considered before. 
  Note that the size of $i_b$ has to be at most the size of $i_{\max}$, since $i_{\max}$ is the largest item, that fits inside $b$.
  As a consequence, no bin is over-packed after this swap since the total size of the items inside the other bin decreases or stays the same.
\end{claimproof}
\fi

\begin{claim}
\label{clm:SmallItemProperty3}
  Consider an optimal solution where each partially filled bin contains exactly two items.
  Let $i_s$ be the smallest large item and $i_\ell$ be the largest one and let them fit together inside a partially filled bin.
  Then there exists an optimal solution, where $i_s$ is positioned inside a partially filled bin, together with the largest large item, that does fit additionally.
\end{claim}
\iflongversion 
\begin{claimproof}
  Consider the optimal solution and the position of the smallest large item $i_s$. 
  If $i_s$ is positioned inside a partially filled bin, we can swap the additional large item with the largest item that fits together with $i_s$ into this bin.
  Since this swap replaces an item inside one other bin with a smaller item, the total size of the items inside this bin decreases and hence no bin is be over-packed.
	
  If $i_s$ is not positioned inside a partially filled bin, then it fits together with the other large item it is currently paired with into a partially filled bin since we assumed that $i_s$ even fits together with the largest item $i_l$ into a partially filled bin. 
  We swap this pair with the two large items of one (arbitrary) fitting partially filled bin. 
  After this swap no bin is over-packed since the other two large items fit in a partially filled bin and hence they fit inside an empty bin as well. 
  Finally, we swap the item that is currently paired with the small item with the largest item that fits inside this bin together with $i_s$.
  As seen above, after this swap there is no bin that is over-packed.
\end{claimproof}
\fi

\begin{claim}
\label{clm:SmallItemProperty4}
  Consider an instance $I$, where the largest large item $i_\ell$ does not fit together with the smallest large item $i_s$ inside any partially filled bin and there is an optimal solution, where all partially filled bins contain exactly two large items.
  Then there is an optimal solution which places $i_\ell$ together with the largest fitting large item inside one bin, or $i_\ell$ is placed alone inside a bin, if there no large item fits together with $i_\ell$ inside one bin.
\end{claim}
\iflongversion 
\begin{claimproof}
  Consider an optimal solution for the given instance $I$, where each partially filled bin contains exactly two large items and $i_\ell$ and $i_s$ do not fit together inside a partially filled bin.
  Consider the bin $b_1$ containing the item $i_\ell$. 
  Obviously $i_\ell$ is not contained inside a partially filled bin, since it does not fit together with the smallest large item inside a bin and hence it cannot fit together with an other large item inside a partially filled bin.
  If there consider the largest item $i$ that does fit together with $i_\ell$ inside one bin and let $b_2$ be the bin containing this item.
  We can swap the item $i_+$ (if existent) that is currently placed together with $i_\ell$ with the item $i$. 
  Since the item $i_+$ has at most the size of the item $i$ the bin $b_2$ is not over-packed by this step.
  On the other hand, since $i_\ell$ and $i$ fit together inside a bin, and there is no small item inside $b_1$ this bin is not over-packed as well.

  If there is no large item that fits together with $b_1$ inside one bin, and $b_1$ does not contain any small items, $b_1$ is contained alone inside its bin in this case.
\end{claimproof}
\fi

\vspace{-2em}
\subparagraph*{The Complete Algorithm.}
In the first step of the algorithm, we sort the items regarding their sizes in $\cO(n \log(n))$.
Next, we guess the distribution of the small items. 
Since there are at most $k$ small items, there are at most $\cO(k!)$ possible guesses.
We call the bins containing small items partially filled bins.
There are at most $k$ of these bins.

Then, we guess a bin $b_1$ that does not contain any additional large item.  
All the partially filled bins, containing small items with a larger total size than $b_1$ do not contain any large item as well, see  \Cref{clm:SmallItemProperty1}.
Thus we can discard them from the following considerations.
There are at most $k$ possibilities for the guess of $b_1$.

Now, we guess which of remaining partially filled bins only contain one large item. 
There are at most $\cO(2^k)$ possibilities.
We consider all partially filled bins for which we guessed that they only contain one large item in any order and pair them with the largest fitting item.
By \Cref{clm:SmallItemProperty1}, we know that an optimal packing with this structure exists.
Afterwards, we discard these bins from the following considerations.

It remains to pack the residual large items.
Each residual, partially filled bin contains exactly two large items in the optimal solution, otherwise the guess was wrong. 
To place the correct large item, we proceed as follows:
Iterate through the large items in non-ascending order regarding their sizes.
Let $i_\ell$ be the currently considered item.
Further, let $i_s$ be the smallest large item from the set of large items that still need to be placed. 
Depending on the relation between $i_\ell$ and $i_s$, we place at least one of these two items inside a bin.
For the first case, it holds that $i_\ell$ does not fit together with $i_s$ inside a partially filled bin.
Then, we place $i_\ell$ together with the largest fitting item $i$ from the set of large items that are not already placed inside one empty bin or place it alone inside an empty bin if such an item does not exist. 
The item $i$ can be found, or its non-existence be proved, in time $\cO(\log(n))$.
For the second case, it holds that $i_\ell$ together with $i_s$ does fit inside one partially filled bin. Then, we guess which partially filled bin contains $i_s$ and place it inside this bin together with the largest unplaced item that fits inside this bin. 
The largest fitting item can be found in time $\cO(\log(n))$, and there are at most $\cO(k!)$ possible guesses total.

In the following, we argue that in both cases there exists an optimal solution where the items are placed exactly as the algorithm does assuming all the guesses are correct. 
When all the previous steps are correct, we can consider the residual set of items as a new instance, where there exists an optimal solution, where all partially filled bins contain exactly two large items (and we already know the correct distribution of small items).
For this new instance we fill one bin correctly due to \Cref{clm:SmallItemProperty4} in Case 1. 
Since this bin is filled correctly with respect to an existing optimal solution, we again can consider the residual set of items as an independent instance that needs solving.
On the other hand in Case 2, we know by \Cref{clm:SmallItemProperty3}, that there exists an optimal solution for this reduced instance where $i_s$ is placed together with the largest fitting large item inside one partially filled bin. 
If we guess this bin correctly, we have filled one bin correctly with regard to the considered instance. 
Hence when reducing the considered instance to the residual set of items (without this just filled bin) there exists an optimal solution for this instance with exactly one less bin. 

After placing all the large items, we compare the obtained solution with the so far best solution, save it if it uses the smallest number of bins so far, and backtrack to the last decision.
Since it iterates all possible guesses, this algorithm generates an optimal packing and its run time is bounded by $\cO((k!)^2 \cdot k \cdot 2^k \cdot n  \log(n))$. 

\vspace{-1em}
\section{Conclusion and Further Work}\label{section:conclusion}
We provided a randomized algorithm with one-sided error to identify perfect over-the-rainbow matchings.
Via reductions to this problem, we obtained randomized $4^k\cdot k!\cdot n^{O(1)}$-time algorithms for the vector versions of \acl{BP}, \acl{MKP}, and \acl{BC} parameterized by the number $k$ of \emph{small items.}
We believe that studying this parameter is a natural step towards the investigation of the stronger parameterizations by the number of distinct item types.
In that setting, the number of small items can then be large---however, there are only few small-item-types and, thus, we may hope to adapt some of the techniques developed in this article to this setting.

As a working horse we used a randomized algorithm to find conjoining matchings.
As mentioned by Marx and Pilipczuk, it seems challenging to find such matchings by a deterministic fixed-parameter algorithm.
Alternatively, we could search directly for deterministic algorithms for the problems as presented in this article. We present such an algorithm for \acl{pBP}; however, the techniques used in its design do not seem to generalize to the vector version.

\bibliography{bib}

\end{document}